\documentclass{new_tlp}
\usepackage{mathptmx}

\newcommand\bcmdtab{\noindent\bgroup\tabcolsep=0pt%
  \begin{tabular}{@{}p{10pc}@{}p{20pc}@{}}}
\newcommand\ecmdtab{\end{tabular}\egroup}

\usepackage{amsmath}
\usepackage{graphicx}
\usepackage{multirow}
\usepackage{xspace}
\usepackage{comment}
\usepackage{hyperref}
\usepackage[mathscr]{euscript}

\def\aopl{$\mathscr{AOPL}$\xspace}

\def\al{$\mathscr{AL}_d$\xspace}
\def\p{$\mathscr{P}$\xspace}
\def\t{$\mathscr{T}$\xspace}
\def\lp{$lp(\mathscr{P}, \sigma)$\xspace}
\def\Ps{$\textbf{P}(\sigma)$\xspace}
\def\apia{$\mathscr{APIA}$\xspace}
\def\reilp{$rei\_lp(\mathscr{P})$\xspace}
\def\reilps{$rei\_lp(\mathscr{P}, \sigma)$\xspace}

\def\ev{$\langle \sigma, a \rangle$\xspace}

\newtheorem{example}{Example}
\newtheorem{definition}{Definition}
\newtheorem{proposition}{Proposition}
\newtheorem{theorem}{Theorem}

  \title[Framework for Policy Refinement]
        {An ASP Framework for the Refinement of Authorization and Obligation Policies}

  \author[D. Inclezan]
         {DANIELA INCLEZAN\\
         Miami University, Oxford, OH, USA\\
         \email{inclezd@miamioh.edu}}

\jdate{May 2023}
\pubyear{2023}
\pagerange{\pageref{firstpage}--\pageref{lastpage}}
\begin{document}

\label{firstpage}

\maketitle

\begin{abstract}
This paper introduces a framework for assisting policy authors in refining and improving their policies. In particular, we focus on authorization and obligation policies that can be encoded in  Gelfond and Lobo's \aopl language for policy specification. We propose a framework that detects the statements that make a policy inconsistent, underspecified, or ambiguous with respect to an action being executed in a given state. We also give attention to issues that arise at the intersection of authorization and obligation policies, for instance when the policy requires an unauthorized action to be executed. The framework is encoded in Answer Set Programming.

\end{abstract}


%%%%%%%%%%%%%%%%%%%%%%%%%%%%%%%%%%%%%%%%%

\section{Introduction}

This paper introduces a framework for assisting policy authors in refining and improving the policies they elaborate. Here, by a {\em policy} we mean {\em a collection of statements that describe the permissions and obligations related to an agent's actions.}

In particular, we focus on authorization and obligation policies that can be encoded in the policy specification language \aopl by Gelfond and Lobo \citeyear{gl08}. \aopl allows an author to specify policies for an autonomous agent acting in a changing environment. A description of the dynamic domain in terms of sorts of the domain, relevant fluents, and actions is assumed to be available to the policy writer. Policy rules of \aopl may be of two kinds: \emph{authorization} rules specifying what actions are permitted/ not permitted and in which situations, and \emph{obligation} rules indicating what actions an agent must perform or not perform under certain conditions. Rules can either be strict or defeasible, and preferences between defeasible rules can be set by the policy author. The semantics of \aopl is defined via a translation into Answer Set Programming (ASP) \cite{gl91}.
Gelfond and Lobo define policy properties such as consistency and categoricity. However, there is a gap in analyzing what happens at the intersection between authorization and obligation policies, for instance when a policy requires an unauthorized action to be executed, which is called a {\em modality conflict} by Craven et al. \citeyear{clmrlb09}.

We propose a framework that detects the rules that make a policy inconsistent, underspecified, or ambiguous with respect to
an action and a given state. The goal is to notify the policy author about the natural language statements in the policy that may be causing an issue and explain why that is the case for the particular action and state.

Given rapid advancements in AI in the past years, the importance of setting and enforcing policies on intelligent agents has become paramount. At the same time, policy specifications can become large and intricate. Thus, assisting policy authors and knowledge engineers with policy refinement by automatically detecting %inconsistencies, underspecification, and ambiguity 
issues %in their policies, 
in a provably correct way and highlighting conflicting policy statements is of great importance.% to the field.

The contributions of our work are as follows:
\begin{itemize}
    \item We define a {\bf new translation} of \aopl policies into ASP by reifying policy rules.
    \item We formally {\bf define issues} that may arise in \aopl policies and describe how to {\bf detect the causing policy statements}, using the reified ASP translation.
    \item We define means for {\bf explaining the root causes for issues} like inconsistency, underspecification, ambiguity, and modality conflicts.
    %\item We 
\end{itemize}

In what follows, we provide a short description of language \aopl in Section \ref{aopl} and give a motivating example in Section \ref{motivation}. We describe our new translation of \aopl policies into ASP in Section \ref{reify} and introduce our framework in Section \ref{analysis}. We discuss related work in Section \ref{related_work} and end with conclusions and future work.

%%%%%%%%%%%%%%%%%%%%%%%%%%%%%%%%%%%%%%%%%

\section{Background: Language \aopl}
\label{aopl}
Let us now briefly present the \aopl language.
We direct the unfamiliar reader to outside resources on ASP \cite{gl91,mt99} and
action language \al \cite{gi13,gk14}, which are also relevant to this work.

Gelfond and Lobo \citeyear{gl08}\footnote{Available at https://www.depts.ttu.edu/cs/research/documents/44.pdf} introduced the Authorization and Obligation Policy Language \aopl for specifying policies for an intelligent agent acting in a dynamic environment. A policy is a collection of authorization and obligation statements, which we simply call authorizations and obligations, respectively. An {\em authorization} indicates whether an agent's action is permitted or not, and under which conditions. 
An {\em obligation} describes whether an agent is obligated or not obligated to perform a specific action under certain conditions. 
An \aopl policy works in conjunction with a dynamic system description of the agent's environment written in an action language such as \al. 
The signature of the dynamic system description includes predicates denoting {\em sorts} for the elements in the domain;
{\em fluents} (i.e., properties of the domain that may be changed by actions); and {\em actions}.
As in \al, we consider dynamic systems that can be represented by a directed graph, called a {\em transition diagram}, containing a finite number of nodes representing physically possible states of the dynamic domain. A state is a complete and consistent set of fluent literals. Arcs in the transition diagram are labeled by action atoms (shortly {\em actions}) that take the system from one state to another. Actions can be elementary or compound, where a compound action is a set of elementary actions executed simultaneously. %Elements of the domain are sorted and they belong to at least one sort. 

{\bf The signature of an \aopl policy} includes the signature of the associated dynamic system and additional predicates $permitted$ for authorizations,
$obl$ for obligations, and \textit{prefer} for establishing preferences between authorizations or obligations. A \textit{prefer} atom is created from the predicate \textit{prefer}; similarly for $permitted$ and $obl$ atoms.

\begin{definition}[Policy]
\label{def1}
An \aopl \textbf{policy} \p is a finite collection of statements of the form:
\begin{subequations} \label{eq2}
    \begin{align}
            & \ \ permitted\left(e\right) & \textbf{ if } \ cond \label{eq1_1}\\[-0.3em]
            & \neg permitted\left(e\right) & \textbf{ if } \ cond \label{eq1_2}\\[-0.3em]
		& \ \ obl\left(h\right) & \textbf{ if } \ cond \label{eq1_3}\\[-0.3em]
            & \neg obl\left(h\right) & \textbf{ if } \ cond\label{eq1_4}\\[-0.3em]
        d: \textbf{normally } & \ \ permitted(e) & \textbf{ if } \ cond \label{eq2_1}\\[-0.3em]
        d: \textbf{normally } & \neg permitted(e) & \textbf{ if } \ cond \label{eq2_2}\\[-0.3em]
        d: \textbf{normally } & \ \ obl(h) & \textbf{ if } \ cond \label{eq2_3}\\[-0.3em]
        d: \textbf{normally } & \neg obl(h) & \textbf{ if } 
        \ cond \label{eq2_4}\\[-0.3em]
	    	& \ \ \mbox{\textit{prefer}}(d_i, d_j) & \label{eq2_5}
    \end{align}
\end{subequations}
where $e$ is an elementary action; $h$ is a happening (i.e., an elementary action or its negation\footnote{If $obl(\neg e)$ is true, 
then the agent must not execute $e$.}); 
$cond$ is a (possibly empty) collection of atoms of the signature, except for atoms containing the predicate \textit{prefer}; $d$ appearing in (\ref{eq2_1})-(\ref{eq2_4}) denotes a defeasible rule label; and $d_i$, $d_j$ in (\ref{eq2_5}) refer to two distinct rule labels from \p.
Rules (\ref{eq1_1})-(\ref{eq1_4}) encode {\em strict} policy statements, while rules (\ref{eq2_1})-(\ref{eq2_4}) encode {\em defeasible} statements (i.e., statements that may have exceptions). Rule (\ref{eq2_5}) captures {\em priorities} between defeasible statements.

In deontic terms, rules (\ref{eq1_1}) and (\ref{eq2_1}) denote \textit{permissions}; rules (\ref{eq1_2}) and (\ref{eq2_2}) denote \textit{prohibitions};
rules
 (\ref{eq1_3}) and (\ref{eq2_3}) denote \textit{obligations}; and
 rules  (\ref{eq1_4}) and (\ref{eq2_4}) denote \textit{dispensations}.
\end{definition}

The\textbf{ semantics} of an \aopl policy determine a mapping \Ps from states of a transition diagram \t into a collection of $permitted$ and $obl$ literals. To formally describe the semantics of \aopl, a translation of a policy and transition diagram into ASP is defined.

\begin{definition}[ASP Translation of a Policy and State]
\label{def2}
{
The translation $lp$ is defined as:
\begin{itemize}
    \item If $x$ is a fluent literal, action literal, $permitted$, or $obl$ literal, then $lp(x) =_{def} x$ .
    \item If $L$ is a set of literals, then $lp(L) =_{def} \{lp(l) : l \in L\}$
    \item If $r = ``l \ \textbf{if} \ cond"$ is a strict rule like the ones in (\ref{eq1_1})-(\ref{eq1_4}), then $lp(r) =_{def} lp(l) \leftarrow lp(cond)$
    \item If $r$ is a defeasible rule like (\ref{eq2_1}) or (\ref{eq2_2}), or a preference rule ``\textit{prefer}$(d_i, d_j)$" like the one in (\ref{eq2_5}), then $lp(r)$ is obtained using standard ASP techniques for encoding defaults, as shown in equations (\ref{eq3_1}), (\ref{eq3_2}), and (\ref{eq3_3}) respectively:
\begin{subequations} \label{eq3}
    \begin{align}
    permitted(e) & \leftarrow \ lp(cond), \mbox{not}\ ab(d), \mbox{not}\ \neg permitted(e) \label{eq3_1}\\
    \neg permitted(e) & \leftarrow \ lp(cond), \mbox{not}\ ab(d), \mbox{not}\ permitted(e) \label{eq3_2}\\
    ab(d_j) & \leftarrow \ lp(cond_i) \label{eq3_3}
    \end{align}
\end{subequations}
where $cond_i$ is the condition of $d_i$.
Similarly for defeasible obligations (\ref{eq2_3}) and (\ref{eq2_4}).
    \item If \p is a policy, then $lp(\mathscr{P}) =_{def} \{lp(st) : st \in \mathscr{P}\}$.
    \item If $\mathscr{P}$ is a policy and $\sigma$ is a state of the (transition diagram associated with the) dynamic system description $\mathscr{T}$, $$lp(\mathscr{P}, \sigma) =_{def} lp(\mathscr{P}) \cup lp(\sigma)$$
\end{itemize}
}
\end{definition}

Properties of an \aopl policy \p are defined in terms of the answer sets of the logic program $lp(\mathscr{P}, \sigma)$ expanded with appropriate rules.

The following definitions by Gelfond and Lobo are relevant to our work (original definition numbers in parenthesis). In what follows $a$ denotes a (possibly) compound action (i.e., a set of simultaneously executed elementary actions), while $e$ refers to an elementary action. An event \ev is a pair consisting of a state $\sigma$ and a (possibly) compound action $a$ executed in $\sigma$.\footnote{In policy analysis, we want to encompass all possible events, i.e., pairs consisting of a physically possible state $\sigma$ and physically executable action $a$ in $\sigma$.}

%%%%%%
\begin{definition}[Consistency -- Def. 3]
\label{def:consistency}
A policy \p for \t is called {\em consistent} if for every state $\sigma$ of \t, the logic program \lp is consistent, i.e., it has an answer set.
\end{definition}

\begin{comment}
\begin{definition}[\Ps for authorization - Def. 4 in the original paper]
Let \p be a consistent authorization policy for \t. Then:

$permitted(e) \in$ \Ps iff the logic program \lp entails $permitted(e)$

$\neg permitted(e) \in$ \Ps iff the logic program \lp entails $\neg permitted(e)$
\end{definition}
\end{comment}

%%%%%%%
\vspace{-5pt}
\begin{definition}[Policy Compliance for Authorizations -- Defs. 4 and 5]
\label{def:auth_compliance}

$\bullet\ $ An event \ev is {\em strongly compliant} with authorization policy \p if for every $e \in a$ we have that $permitted(e) \in$ \Ps (i.e., the logic program \lp entails $permitted(e)$).

\noindent
$\bullet\ $ An event \ev is {\em weakly compliant} with authorization policy \p if for every $e \in a$ we have that $\neg permitted(e) \notin$ \Ps (i.e., the logic program \lp does not entail $\neg permitted(e)$).

\noindent
$\bullet\ $ An event \ev is {\em non-compliant} with authorization policy \p if for every $e \in a$ we have that $\neg permitted(e) \in$ \Ps (i.e., the logic program \lp entails $\neg permitted(e)$).
\end{definition}

%%%%%%%%%%

\vspace{-5pt}
\begin{definition}[Policy Compliance for Obligations -- Def. 9]
\label{def:obl_compliance}
An event \ev is {\em compliant} with obligation policy \p if 

\noindent
$\bullet\ $ For every $obl(e) \in$ \Ps we have that $e \in a$, and

\noindent
$\bullet\ $ For every $obl(\neg e) \in$ \Ps we have that $e \notin a$.
\end{definition}

%%%%%%%%%%%%
\vspace{-5pt}
\begin{definition}[Categoricity -- Def. 6]
\label{def:categoricity}
A policy \p for \t is called {\em categorical} if for every state $\sigma$ of \t the logic program \lp is categorical, i.e., has exactly one answer set.
\end{definition}

Note that \aopl does not discuss interactions between authorizations and obligations referring to the same action, for instance situations when both $obl(e)$ and $\neg permitted(e)$ are part of the answer set of \lp for some state $\sigma$.

%%%%%%%%%%%%%%%%%%%%%%%%%%%%%%%%%%%%%%%%%

\section{Motivating Example}
\label{motivation}

To illustrate the policy refinement process that we want to facilitate, let's consider an example provided by Gelfond and Lobo, expanded with an additional rule 4):

\begin{example}[Authorization Policy Example]
\label{ex0}
1) A military officer is not allowed to command a mission they authorized.

\noindent
 2) A colonel is allowed to command a mission they authorized.
 
 \noindent
3) A military observer can never authorize a mission.

 \noindent
4) A military officer must command a mission if ordered by their superior to do so.
\end{example}

Before discussing the encoding of this policy, let us assume that the description of this domain includes actions $assume\_comm(C, M)$ and $authorize\_comm(C, M)$; fluents $autorized(C, M)$ and $ordered\_by\_superior(C, M)$; and sorts $colonel(C)$ and $observer(C)$, where $C$ is a commander and $M$ is a mission.

In the English description of the policy in Example \ref{ex0}, note that statements 1) and 2) are phrased as strict rules and thus an automated translation process into \aopl would produce the policy:
\begin{subequations} 
\label{ex1}
    \begin{align}		
    \neg permitted(assume\_comm(C, M)) & \textbf{ if } 
        authorized(C, M) \label{ex1_1}\\[-0.2em]
    permitted(assume\_comm(C, M)) & \textbf{ if } colonel(C) 
        \label{ex1_2}\\[-0.2em]
    \neg permitted(authorize\_comm(C, M)) & \textbf{ if } 
        observer(C) \label{ex1_3}\\[-0.2em]
    obl(assume\_comm(C, M)) & \textbf{ if } 
        ordered\_by\_superior(C, M) \label{ex1_4}
    \end{align}
\end{subequations}

Such a policy is inconsistent in a state in which $authorized(c, m)$ and $colonel(c)$ both hold, due to rules (\ref{ex1_1}) and (\ref{ex1_2}). Gelfond and Lobo indicate that ``[s]ince the nature of the first two authorization policy statements of our example are contradictory we naturally assume them to be defeasible" and replace the encoding in rules (\ref{ex1_1}) and (\ref{ex1_2}) with
%\begin{equation} 
%\label{ex2}
$$
    \begin{array}{ll}		
    d_1(C, M)\ : & {\bf normally}  \ 
         \neg permitted(assume\_comm(C, M)) \ \textbf{ if } authorized(C, M) \\
    d_2(C, M)\ : & {\bf normally} \ \ \
         permitted(assume\_comm(C, M)) \ \textbf{ if } colonel(C) \\
    & prefer(d_2(C, M), d_1(C, M)) 
    \end{array}
$$    
%\end{equation}
while leaving rule (\ref{ex1_3}) as it is. Rule (\ref{ex1_4}) is unaffected, as it corresponds to the new policy statement 4) that we added to the original example by Gelfond and Lobo, to illustrate obligations.

\begin{comment}
\textcolor{red}
{Show the new English policy?}
\end{comment}

This approach has several drawbacks: (a) it puts the burden on the knowledge engineer, who may have a more limited knowledge of the domain than the policy author and thus may make false assumptions; (b) it does not scale for large and intricate policies; and (c) it would be difficult to automate. Instead, we propose a framework that detects inconsistencies like the one above, alerts the policy author of the conflicting policy statements and the conditions that cause them, and allows the policy author to refine the policy (with options for refinement possibly suggested, in the future). In particular, for Example \ref{ex0}, we expect our framework to indicate that statements 1) and 2) are in conflict in a state in which both $colonel(c)$ and $authorized(c, m)$ are true. Similarly, our framework should flag the contradiction between the obligation in rule 4) and rule 1) in a state in which $authorized(c, m)$ and $ordered\_by\_superior(c, m)$ both hold.

To achieve this goal, we modify \aopl by introducing labels for all rules (including strict and preference rules)\footnote{As in the original \aopl language, preferences can be defined only between pairs of defeasible rules.} and connecting rules of \aopl with natural language statements of the original policy via a new predicate $text$ as in the following example for the strict policy rule in (\ref{ex1_3}) where $s_1$ is the label for the strict authorization rule:
$$
\begin{array}{ll}
s_1 :& \neg permitted(authorize\_command(C, M)) \textbf{ if }        observer(C) \\
& text(s_1, ``\mbox{A military observer can never authorize a mission.}")
\end{array}
$$
Additionally, we define a different translation of \aopl into ASP, which we will denote by $rei\_lp$ as it {\em reifies} policy rules. We define the $rei\_lp$ translation in the next section.

%%%%%%%%%%%%%%%%%%%%%%%%%%%%%%%%%%%%%%%%%

\section{Reification of Policy Rules}
\label{reify}

The new translation of \aopl into ASP that we propose follows previous methods for the reification of rules in other domains, such as reasoning about prioritized defaults \cite{gs97} or belief states \cite{bgps20}. %As stated above, we assume that each rule has a label $r$.
Similar to the definition of the $lp$ translation function, the signature of \reilp for a policy \p applying in a dynamic domain described by \t contains the sorts, fluents, and actions of \t. To simplify the presentation, we limit ourselves to boolean fluents and use the general syntax  
\begin{equation}\label{gen_rule}
r : [{\bf normally}] \ hd \ {\bf if} \  cond
\end{equation}
to refer to both strict and defeasible, authorization and obligation rules from \p. We use the term  {\em head} of rule $r$ to refer to the $hd$ part in (\ref{gen_rule}), where $hd \in HD$,
$$HD = \bigcup\limits_{e \in E}\{permitted(e), \neg permitted(e), obl(e), obl(\neg e), \neg obl(e), \neg obl(\neg e)\}$$ 
and $E$ is the set of all elementary actions in \t. 
%
\begin{comment}
\textcolor{red}
{What about preference rules - do we need to label them?}
\end{comment}
The signature of \reilp also includes the elements of $HD$ and the following predicates:
\begin{itemize}
\item $rule(r)$ -- where $r$ is a rule label (referred shortly as ``rule" below)
\item $type(r, ty)$ -- where $ty \in \{strict,$ \textit{defeasible, prefer}$\}$ is the type of rule $r$
\item $text(r, t)$ -- to denote that rule $r$ corresponds to %natural language 
policy statement $t$
\item $head(r, hd)$ -- to denote the head $hd$ of rule $r$ 
\item $body(r, b(r))$ -- where $b(r)$ denotes the condition $cond$ in rule $r$ and $b$ is a new function added to the signature of $rei\_lp(\mathscr{P})$
\item $mbr(b(r), l)$ -- for every $l$ in the condition $cond$ of rule $r$
($mbr$ stands for ``member")
%\item $mbr(b(r), neg(f))$ - for every $\neg f$ in the condition $cond$ of rule $r$
\item $ab(r)$ -- for every defeasible rule $r$
\item $holds(x)$ -- where $x$ may be a rule $r$, the head $hd$ of a rule, function $b(r)$ representing the $cond$ of a rule $r$, literal $l$ of \t, or $ab(r)$ from above
\item $opp(r, \overline{hd})$ -- where $r$ is a defeasible rule and  $\overline{hd} \in HD$ ($opp$ stands for ``opposite")
\item \textit{prefer}$(d_1, d_2)$ -- where $d_1$ and $d_2$ are defeasible rule labels
\end{itemize}

The predicate $holds$ helps determine which policy rules are applicable, based on what fluents are true/false in a state and the interactions between policy rules. The predicate $opp(r, \overline{hd})$ indicates that $\overline{hd}$ is the logical complement of $r$'s head $hd$.

The {\bf translation \reilp} consists of facts encoding the rules in \p using the predicates $rule$, $type$, $head$, $mbr$, and \textit{prefer}, as well as the set of policy-independent rules %${\cal R}$ 
below, which define predicates $holds$ and $opp$, where $L$ is a fluent literal, $E$ an elementary action, and $H$ a happening (i.e., an elementary action or its negation). %\footnote{The definition of $holds$ for defeasible rules matches the one in (\ref{eq3_1}) and (\ref{eq3_2}) together.}
$$
\begin{array}{lll}
    body(R, b(R)) & \leftarrow & rule(R)\\
    holds(R) & \leftarrow & type(R, strict), holds(b(R))\\
    holds(R) & \leftarrow & type(R, defeasible), holds(b(R)), \\
             & & opp(R, O), \mbox{not } holds(O), \mbox{not } holds(ab(R))\\          
    holds(B) & \leftarrow & body(R, B), N = \#count\{L : mbr(B, L)\},\\
             & & N = \#count\{L : mbr(B, L), holds(L)\}\\
    holds(ab(R2)) & \leftarrow & \mbox{\textit{prefer}}(R1, R2), holds(b(R1)) \\
    holds(Hd) & \leftarrow &  rule(R), holds(R), head(R, Hd)\\ 
    opp(R, permitted(E)) & \leftarrow & head(R, \neg permitted(E))\\
    opp(R, \neg permitted(E)) & \leftarrow & head(R, permitted(E))\\%\\
        %& \leftarrow & holds(permitted(E)), holds(\neg permitted(E))\\
        %& \leftarrow & holds(obl(E)), holds(\neg obl(E))
    opp(R, obl(H)) & \leftarrow & head(R, \neg obl(H))\\
    opp(R, \neg obl(H)) & \leftarrow & head(R, obl(H))
\end{array}
$$

\vspace{-5pt}
\begin{definition}[Reified ASP Translation of a Policy and State]
\label{def:rei_state}
Given a state $\sigma$ of \t,  \reilps $=_{def} $ \reilp $\cup \ \{ holds(l) : l \in \sigma \ \}$. 
\end{definition}

\noindent
This definition will be used in conducting various policy analysis tasks in Section \ref{analysis}.

\begin{proposition}[Relationship between the Original and Reified ASP Translations]
\label{prop1}
Given a state $\sigma$ of \t, there is a one-to-one correspondence $map : {\cal A} \rightarrow {\cal B}$ between the collection of answer sets ${\cal A}$ of \lp and the collection of answer sets ${\cal B}$ of \reilps such that if $map (A) = B$ then $\forall hd \in HD \cap A, \exists holds(hd) \in B$.
\end{proposition}

%This equivalence between translations will allow us to use the most convenient ASP encoding when conducting policy analysis tasks, as done in Theorem \ref{th1} below.

\begin{comment}
\textcolor{red}{Add an example, but not for policy from Ex. 1 - what to do about vars?}
\end{comment}

%%%%%%%%%%%%%%%%%%%%%%%%%%%%%%%%%%%%%%%%%

\section{Policy Analysis}
\label{analysis}

In what follows, we assume that the $cond$ part of a policy rule cannot include atoms from the set $HD$ (i.e., atoms obtained from predicates $permitted$ and $obl$). We plan to consider more general policies in future work. Lifting this restriction complicates the task of finding explanations beyond the goal of the current work.
%: it requires following chains of authorizations or obligations instead of directly finding explanations in the physical nature of a state.

\subsection{Analyzing Authorization Policies}
\label{sec:auth}
Our goal in analyzing policies is to assist a policy author in refining their policies by indicating to them the rules that cause concern. Thus, when analyzing an authorization policy \p with respect to an elementary action $e$ in a state $\sigma$, we focus on the %following 
tasks:
\begin{itemize}
    \item {\bf Explain the causes of inconsistencies} -- determining the rules that cause a policy to derive both $holds(permitted(e))$ and $holds(\neg permitted(e))$ when using the $rei\_lp$ translation
    \item {\bf Detect and explain underspecification} -- determining whether rules about $e$ exist or not, and, if they exist, explain why they do not fire. (Craven et al. \citeyear{clmrlb09} call this situation {\em coverage gaps}.)
    \item {\bf Detect and explain ambiguity} -- determining whether there are conflicting defeasible rules that produce $holds(permitted(e))$ in some answer sets and $holds(\neg permitted(e))$ in others, and indicating which rules these are.
\end{itemize}

\subsubsection{Inconsistency}

To detect and explain inconsistencies with respect to an elementary action $e$ and state $\sigma$ we introduce the following predicates: 
\begin{itemize}
    \item $inconsistency(e, r_1, r_2)$ -- indicates that the pair of rules $r_1$ and $r_2$ both fire and cause the inconsistency; $r_1$ produces $permitted(e)$ and $r_2$ produces $\neg permitted(e)$
    \item $inconsistency\_expl(e, t_1, t_2)$ -- does the same but indicates the natural language texts of the corresponding policy statements
    \item $inconsistency\_expl\_pos(e, l)$ -- indicates that $l$ is a fluent/static that holds in $\sigma$ and contributes to the inconsistency in a rule that produces $permitted(e)$
    \item $inconsistency\_expl\_neg(e, l)$ -- similar to the previous predicate, but for rules that produce $\neg permitted(e)$
\end{itemize}
We define a logic program \textbf{I} consisting of the rules below:
$$
\begin{array}{lll}
inconsistency(E, R1, R2) & \leftarrow & holds(permitted(E)), holds(\neg permitted(E)),\\
   & & holds(R1), head(R1, permitted(E)), \\
   & & holds(R2), head(R2, \neg permitted(E))
\end{array}
$$
$$
\begin{array}{lll}
inconsistency\_expl(E, T1, T2) & \leftarrow &  inconsistency(E, R1, R2),\\ 
    & & text(R1, T1),  text(R2, T2)\\
inconsistency\_expl\_pos(E, L) & \leftarrow & inconsistency(E, R1, \_), head(R1, permitted(E)), \\
    & &  mbr(b(R1), L), holds(L)\\
inconsistency\_expl\_neg(E, L) & \leftarrow & inconsistency(E, \_, R2), head(R2, \neg permitted(E)),\\
    & &  mbr(b(R2), L), holds(L)
 \end{array}
$$

We restate Definition \ref{def:consistency} in terms of the \reilps translation.
\begin{definition}[Inconsistency Redefined]
An authorization policy \p is {\em inconsistent} with respect to an elementary action $e$ and a state $\sigma$ if the answer set of \reilps $\cup$ \textbf{I} contains $inconsistency(e, r_1, r_2)$ for a pair of rules $r_1$ and $r_2$.
\end{definition}

%Let us also define what is meant by explaining an inconsistency to a policy author.
\vspace{-5pt}
\begin{definition}[Explaining the Causes of Inconsistency]
\label{def:expl_inconsistency}
$\bullet \ $ An explanation for the inconsistency of $e$ in $\sigma$ is the set of pairs of strings

$\{(t_1, t_2) : inconsistency\_expl(e, t_1, t_2) \in $ \reilps $\cup$ \textbf{I}$\}$.

\noindent
$\bullet \ $ 
A fluent literal $l$ contributes positively (or negatively) to the inconsistency of $e$ in $\sigma$ if the answer set of \reilps $\cup$ \textbf{I} contains $inconsistency\_expl\_pos(e, l)$ (or $inconsistency\_expl\_neg(e, l)$, respectively).
\end{definition}

The collection of atoms identified in Definition \ref{def:expl_inconsistency} can be collected from the answer set of \reilps $\cup$ \textit{I} and post-processed to provide more human-friendly output.

\subsubsection{Underspecification}

We define the notion of {\em underspecification} of an elementary action $e$ in a state $\sigma$ as the lack of {\em explicit} information as to whether $e$ is permitted or not permitted in that state, similar to the concept of {\em coverage gap} defined by Craven et al. \citeyear{clmrlb09}. Note that underspecification is different from non-categoricity, where a policy may be ambiguous because $e$ is deemed permitted in some answer sets and not permitted in others. 

% say what categorical with respect to $e$ means
\begin{definition}[Categoricity and Underspecification of an Action in a State]
\label{def:under}
A consistent authorization policy \p is {\em categorical} with respect to an elementary action $e$ and state $\sigma$ if one of the following cases is true:
\begin{enumerate}
\item \reilps entails $holds(permitted(e))$, or

\item \reilps entails $holds(\neg permitted(e))$, or

\item For every answer set $S$ of \reilps, 
 $$\{holds(permitted(e)), holds(\neg permitted(e))\}\cap \ S = \emptyset$$ 
 
 In this last case, we say that $e$ is {\em underspecified} in state $\sigma$.
 \end{enumerate}
\end{definition}

Underspecification is important because it may reflect an oversight from the policy author. If it's unintended, it can have negative consequences in planning domains for example, when an agent may want to choose the most compliant plan and thus actions that are mistakenly underspecified may never be selected.
To test whether an action $e$ is underspecified in a state, we define the set of rules $Check_{und}(e)$ consisting of the set of constraints:  
$$
\begin{array}{ll}
\{ & \leftarrow holds(permitted(e)), \\ 
   & \leftarrow \ holds(\neg permitted(e))\  \ \}
\end{array}
$$
\begin{definition}[Detecting Underspecification]
{
Action $e$ is {\em underspecified} in $\sigma$ if the logic program $lp(\mathscr{P}, \sigma) \cup Check_{und}(e)$ is consistent.
}
\end{definition}

Whenever an elementary action is underspecified in a state, there may be two explanations: (Case 1) the authorization policy contains no rules about $e$, or (Case 2) rules about $e$ exist in the policy but none of them apply in state $\sigma$.
Once we establish that an elementary action is underspecified, we want to explain to the policy author why that is the case.
For the first case, we just want to inform the policy author about the situation. In the second case, we want to report, for each authorization rule about $e$, which fluents make the rule non-applicable. Note that a defeasible rule $r$ with head $hd$ (see (\ref{gen_rule})) cannot be made unapplicable by the complement $\overline{hd}$, as the complement is underivable as well in an underspecified policy. Similarly, a preference rule cannot disable a defeasible rule either, as this can only be the case when the complement $\overline{hd}$ can be inferred.

Let \textbf{U} be the logic program
$$
\begin{array}{lll}
rules\_exist(E) & \leftarrow & head(R, \ permitted(E); \neg permitted(E))\\
underspec\_1(E) & \leftarrow & \mbox{not } rules\_exist(E)\\
underspec\_1\_expl(\mbox{``Case 1"}, E) & \leftarrow & 
  underspec\_1(E)\\

underspec\_2(E) & \leftarrow & rules\_exist(E)\\
underspec\_2(E, R, L) & \leftarrow & underspec\_2(E), rule(R),\\ 
  & & head(R,\  permitted(E); \neg permitted(E)), \\
  & & mbr(b(R), L), \\
  & & \mbox{not } holds(L) \\
underspec\_2\_expl(\mbox{``Case 2"}, E, R, L, T) & \leftarrow & underspec\_2(E, R, L),\\
& & text(R, T)
\end{array}
$$
\begin{definition}[Explaining the Causes of Underspecification]
\label{def:expl_under}
An explanation for the underspecification of $e$ in $\sigma$ is the set of atoms formed by predicates $underspec\_1\_expl$ and $underspec\_2\_expl$ found in the answer set of \reilps $\cup$ \textbf{U}.
\end{definition}

For a more human-friendly explanation, an atom $underspec\_1\_expl(\mbox{``Case 1"}, e)$ in the answer set can be  replaced with an explanation of the form ``There are no authorization rules about $e$" in the post-processing phase. A collection of atoms 
of the form $$\{underspec\_2\_expl(\mbox{``Case 2"}, e, r, l_1, t), \dots, underspec\_2\_expl(\mbox{``Case 2"}, e, r, l_n, t)\}$$ can be replaced with the explanation ``Rule $r$ about action $e$ (stating that ``$t$") is rendered inapplicable by the fact that fluent(s) $l_1, \dots, l_n$ do not hold in this state."

\subsubsection{Ambiguity}

We define ambiguity as the case when the policy allows a choice between $permitted(e)$ and $\neg permitted(e)$. This notion of ambiguity overlaps with that of a non-categorical policy. However, given our assumption that $permitted$ atoms are not included in the condition $cond$ of policy rules, ambiguity is a much more specific case. 
We claim that, if \p is a consistent, non-categorical policy with respect to $e$ in $\sigma$ (see Definition \ref{def:under}), then $holds(permitted(e))$ will be in some answer sets of \reilps and $holds(\neg permitted(e))$ will be in others, but it cannot be the case that an answer set does not contain either.\footnote{Note that, if we lift our restriction and allow $cond$ to contain $permitted$ (or $obl$) atoms, then for a weakly compliant action $e$, there can be a combination of answer sets containing $holds(permitted(e))$ and answer sets not containing neither $holds(permitted(e))$ nor $holds(\neg permitted(e))$, if $cond$ contains a $permitted(e_1)$ atom such that $e_1$ is an action that is ambiguous in $\sigma$.}

The justification is that the body $cond$ of policy rules is fully determined by the {\em unique} values of fluents in $\sigma$. Hence, strict rules either fire or do not. If a strict ruled fired, it would automatically override the defeasible rules with the complementary head, and thus lead either to inconsistency 
(depending on which other strict rules fire) 
or categoricity.
The only source of non-categoricity can be the presence of defeasible rules with complementary heads and satisfied conditions, and which are not overriden by preference rules. %Choosing between one head and its complement can only lead to the situation above.

\begin{definition}[Ambiguity of an Action in a State]
Let \p be a policy that is consistent and non-categorical with respect to elementary action $e$ and state $\sigma$. Let \reilps have $n$ answer sets, out of which $n_p$ answer sets contain $holds(permitted(e))$ and $n_{np}$ contain $holds(\neg permitted(e))$.

\noindent
\p is {\em ambiguous} with respect to $e$ and $\sigma$ if $n \neq n_p$, $n \neq n_{np}$ and $n = n_p + n_{np}$.
\end{definition}

Next, let's describe how we detect ambiguity.
\begin{definition}[Detecting Ambiguity]
Action $e$ is {\em ambiguous} in $\sigma$ if $holds(permitted(e))$ and $holds(\neg permitted(e))$ are not entailed by \reilps and $e$ is not underspecified in $\sigma$.
\end{definition}

Once ambiguity is established, an explanation for ambiguity is needed. To produce it, we define the logic program \textbf{A} consisting of the rules:
$$
\begin{array}{lll}
ambiguous(E, R1, R2) & \leftarrow & defeasible\_rule(R1), head(R1, permitted(E)),\\
   & & defeasible\_rule(R2), head(R2, \neg permitted(E)),\\
   & & holds(b(R1)), holds(b(R2)),\\
   & & \mbox{not } holds(ab(R1)), \mbox{not } holds(ab(R2))\\

ambiguity\_expl(E, T1, T2)& \leftarrow & ambiguous(E, R1, R2), text(R1, T2), text(R2, T2)\\
\end{array}
$$

\begin{definition}[Explaining the Causes of Ambiguity]
\label{def:expl_ambiguity}
An explanation for the ambiguity of $e$ in $\sigma$ is the set of pairs of strings: 
$$\{(t_1, t_2) : ambiguity\_expl(e, t_1, t_2) \in rei\_lp(\mathscr{P}, \sigma) \cup \mbox{ \textbf{A}} \}$$
%A pair of strings $\{t_1, t_2\}$ is an explanation for the ambiguity of $e$ in $\sigma$ if the answer set of \reilps $\cup \ {\cal A}$ contains $ambiguous\_expl(e, t_1, t_2)$.
\end{definition}

\subsubsection{Observation about Strongly vs Weakly Compliant Policies}

Gelfond and Lobo distinguish between actions that are strongly compliant in a state versus weakly compliant (see Definition \ref{def:auth_compliance}). In planning, as shown by Meyer and Inclezan \citeyear{mi21}, it seems reasonable to prefer strongly compliant actions over weakly compliant ones. However, in the theorem below we show that the class of weakly compliant actions includes strongly compliant ones. What we really need for planning purposes is distinguishing between strongly compliant and underspecified actions in a state, so that we can create a preference order between actions. 

\begin{theorem}[Strongly vs Weakly Compliant Actions]
\label{th1}
All elementary actions that are strongly compliant in a state $\sigma$ are also weakly compliant.
\end{theorem}

\begin{proof}
Note that, in this proof, we consider the original $lp$ translation of \aopl, which is equivalent with, but more convenient to use here than, the reified translation $rei\_lp$ as stated in Proposition~\ref{prop1}. According to Definition \ref{def:auth_compliance} borrowed from Gelfond and Lobo's work, elementary action $e$ is strongly compliant with authorization policy \p if $lp(\mathscr{P}, \sigma)$ entails $permitted(e)$, and it is weakly compliant if $lp(\mathscr{P}, \sigma)$ does not entail $\neg permitted(e)$. For consistent policies, the latter condition is obviously true if $e$ is strongly compliant in $\sigma$, as having $permitted(e)$ in every answer set of $lp(\mathscr{P}, \sigma)$ implies that $\neg permitted(e)$ must be absent from each such answer set. If the policy is inconsistent, the theorem is vacuously true.
\end{proof}

Given our assumption that $permitted$ (and $obl$) atoms cannot appear in the $cond$ part of policy rules, an elementary action $e$ can only be either strongly compliant in $\sigma$ or underspecified.%\footnote{If we lift this restriction, then it can also be the case that $permitted(e)$ is in some answer sets, while other answer sets do not contain neither $permitted(e)$ nor $\neg permitted(e)$ (but $e$ is still weakly compliant). This can be to due to rules about $e$ containing $permitted(e_1)$ in condition $cond$, where the policy is not categorical with respect to $e_1$ in $\sigma$, i.e., $permitted(e_1)$ may be true in some answer sets and not in others.}
\footnote{If this restriction is lifted and the condition $cond$ of a policy rule for $e$ contains a $permitted(e_1)$ atom, such that $e_1$ is {\em ambiguous}, then $e$ may be a weakly compliant action because it will be permitted in some answer sets and unknown in others.}

We formulate the following proposition, which is useful in creating an ordering of actions based on compliance (relevant in planning). 
\begin{proposition}[Properties of Authorization Policies]
$\bullet\ $ If condition $cond$ of authorization rules is not allowed to contain $permitted$ (or $obl$) atoms and \p is {\em categorical} with respect to $e$ and $\sigma$, then $e$ is either strongly compliant, non-compliant, or underspecified in $\sigma$.

\noindent
$\bullet\ $ If condition $cond$ of authorization rules is not allowed to contain $permitted$ (or $obl$) atoms and \p is {\em non-categorical} with respect to $e$ and $\sigma$, then $e$ is neither strongly compliant nor non-compliant; it may be either underspecified or ambiguous. % (weakly compliant).
\end{proposition}
% Is this needed here?

%%%%%%%%%%%%%%%%%%%%%%%%%%%%%%%%%%%%%%%%%

\subsection{Obligation Policy Analysis}
\label{obl}

The techniques from Section \ref{sec:auth} for determining rules that create inconsistencies, underspecification, and ambiguity with respect to an  elementary action and a state can be easily adapted to obligation policies as well. Obligation policies apply to {\em happenings}, which are actions or their negations. For instance, given an elementary action $e$, the following literals are part of the signature of the policy: $obl(e)$, $obl(\neg e)$, $\neg obl(e)$, and $\neg obl(\neg e)$. Inconsistencies between $obl(e)$ and $\neg obl(e)$ on one hand, or between $obl(\neg e)$ and $\neg obl(\neg e)$ are easy to detect. However, there are additional incongruencies that may decrease the quality of a policy, and we may want to alert the policy writer about them as well.

For instance, if a policy \p entails both $obl(e)$ and $obl(\neg e)$ in a state $\sigma$, then any event $\langle \sigma, a\rangle$ will be non-compliant, no matter whether $e \in a$ or $e \notin a$. In actuality, it means that the policy does not allow for the agent to be compliant with respect to obligations in state $\sigma$. Thus the notions of inconsistency and ambiguity should be adapted or expanded to include this situation. We propose the following definition:

\begin{definition}[Conflicting Policy]
Given a consistent obligation policy \p, a state $\sigma$ and an elementary action $e$, we call \p a {\em conflicting obligation policy} with respect to $\sigma$ and $e$ if the logic program $rei\_lp(\mathscr{P}, \sigma)$ entails both $holds(obl(e))$ and $holds(obl(\neg e))$. 
\end{definition}

Explanations for conflicting obligation policies can be found using techniques similar to the ones in Section \ref{sec:auth}.
%%%%%%%%%%%%%%%%%%%%%%%%%%%%%%%%%%%%%%%%%

\subsection{The Intersection between Authorization and Obligation Policies}
\label{intersect}

When combining an authorization policy with an obligation policy, there are a few cases that, while not necessarily inconsistent, certainly seem to require non-compliant behavior from the agent. This is especially the case when an event $\langle \sigma, e\rangle$ is strongly-compliant with the obligation policy but non-compliant with the authorization policy (i.e., in terms of the original translation $lp$ of \aopl into ASP, $lp(\mathscr{P}, \sigma)$ entails both $obl(e)$ and $\neg permitted(e)$ according to Definitions~\ref{def:obl_compliance} and \ref{def:auth_compliance}).
Other situations that may require the policy authors' attention, though to a lesser degree, are when an action is permitted but the agent is obligated not to execute it (i.e., $lp(\mathscr{P}, \sigma)$ entails both $permitted(e)$ and $obl(\neg e)$) or when the agent is obligated to execute an action that is underspecified in that state. We indicate the level of urgency of each of these situations by adding a number from 1 to 3, with 1 being the most needing of re-consideration and 3 being the least urgent. 

Once it has been established that the policies are strongly compliant, non-compliant, or underspecified with respect to the state and elementary action, the following ASP rules determine which policy rules need to be re-visited.
$$
\begin{array}{lll}
require\_cons(E, R1, R2, 1) & \leftarrow & holds(R1), head(R1, obl(E)), \\
  & & holds(R2), head(R2, \neg permitted(E))\\
require\_cons(E, R1, R2, 2) & \leftarrow & holds(R1), head(R1, obl(\neg E)), \\
  & & holds(R2), head(R2, permitted(E))\\
require\_cons(E, R1, R2, 3) & \leftarrow & holds(R1), head(R1, obl(E)), \\
  & & \mbox{not } holds(permitted(E)),\\
  & & \mbox{not } holds(\neg permitted(E))\\

require\_cons\_expl(E, T1, T2, N)
    & \leftarrow & require\_cons(E, R1, R2, N), \\
    & & text(R1, T1), text(R2, T2).
\end{array}
$$
\begin{comment}
\textcolor{red}
{Also, second type of coverage gap analysis according to \cite{clmrlb09} in which there is no explicit info about authorization of $e$, but there is $obl(e)$}
\end{comment}

%%%%%%%%%%%%%%%%%%%%%%%%%%%%%%%%%%%%%%%%%

\section{Related Work}
\label{related_work}
Meyer and Inclezan \citeyear{mi21} developed an architecture for policy-aware intentional agents (\apia) by leveraging Blount et al.'s theory of intentions \citeyear{BlountGB15}. An agent's behavior was ensured to be compliant with authorization and obligation policies specified in \aopl and translated into ASP. Meyer and Inclezan's work first highlighted the issues that may arise at the intersection between \aopl authorization and  obligation and  policies. In the \apia architecture, conflicts of this nature were resolved by modifying the policy's ASP encoding to state that such conflicts render a policy inconsistent. In the current work our intention is to alert policy authors about such situations and provide them with the opportunity to decide which policy statements to modify in order to restore consistency. Additionally, in the current work we delve deeper into the tasks associated with policy analysis and look at underspecification and ambiguity as well. We also focus on providing explanations as to why such issues arise.

Craven et al.'s work \citeyear{clmrlb09} is the closest to ours in its intent. The authors define language $\mathscr{L}$ for policy specification and include both authorization and obligation policies. They define a solid set of tasks that an automated analysis of a policy should accomplish, such as discovering {\em modality conflicts} and {\em coverage gaps}, which we target in our work as well. 
Their research assumes that the underlying dynamic domain is specified in Event Calculus \cite{ks89}. Explanations are found via an abductive constraint logic programming proof procedure. Given the absence of a comparison between languages $\mathscr{L}$ and \aopl, it is important to study the problem of policy analysis with respect to language \aopl as well. \aopl has clear advantages, including its ability to express defeasible policies and preferences between policies. Moreover, \aopl can be seamlessly integrated with ASP-based dynamic system descriptions, as different properties of \aopl policies can be checked by finding the answer sets of an ASP program. This would allow coupling policies with system descriptions specified in higher-level action languages that translate into ASP, such as the modular action language $\mathscr{ALM}$ \cite{InclezanG16}, and associated libraries about action and change \cite{di16,di19}.

Other research on policy modeling or analysis using ASP exists, but the goals tend to be different from ours. Corapi et al. \citeyear{crdps11} use inductive logic programming and software engineering-inspired processes to assist policy authors with policy refinement. In their work, refinement suggestions are provided, but this process is driven by use cases that need to be manually created. As a result, the quality of the resulting policy depends on the quality and coverage of the use cases that are provided as an input. In turn, our approach is meant to be more comprehensive and transparent, as it is guided by the policy rules themselves. Another work that uses ASP for policy modeling is that of De Vos et al. \citeyear{dkps19}. Their work encompasses the same types of policies as \aopl, but their focus is on compliance checking and providing explanations for the compliance or non-compliance of events. In contrast, we focus on policy analysis, not compliance checking; our explanations highlight potential problems with a policy and indicate statements that need to be refined. Pellegrini et al. \citeyear{phspfmtpks19} present a framework called DALICC for comparing and resolving compatibility issues with licenses. The goal of their framework is more narrow than ours in the sense that it only focuses on licences and not normative statements in general. For a survey on other policy analysis methods and tools, not necessarily ASP-related, we direct the reader to the paper by Jabal et al. \citeyear{jdbmcvrw19}.

%A survey on policy analysis methods and tools by \cite{jdbmcvrw19} seems to indicate a lack of such artifacts based on ASP. We intend our \aopl and ASP based framework to fill this gap.

In general, in the policy specification and analysis community, there is a intense focus on access control policies, which may involve the Role-Based Access Control (RBAC) model outlined by Ferraiolo et al. \citeyear{fsgkc01}; the Attribute-Based Access Control Model (ABAC) explored for instance by Davari and Zulkernine \citeyear{dz21} and Xu et al. \citeyear{xwps16}; or the Category-based Access Control Model explored by Alves and  Fern\'{a}ndez \citeyear{AlvesF15}. A secondary focus falls on policies for the management of computer systems. In contrast, \aopl is more general and could be used to represent social norms, for example.

Finally, our work touches upon explainability and finding the causes of issues encountered in \aopl policies. To find even deeper causes that reside in the inner-workings of the dynamic system, we can leverage existing work on explainability in reasoning about action and change, such as the research by LeBlanc et al. \citeyear{LeBlancBV19}; planning domains, as in the work by Vasileiou et al. \citeyear{VYSKCM22}; or logic programming in general, including research by Fandinno and Schultz \citeyear{FandinnoS19} or Cabalar et al. \citeyear{cfm20}. 

%%%%%%%%%%%%%%%%%%%%%%%%%%%%%%%%%%%%%%%%%

\section{Conclusions}

In this paper we introduced a framework for analyzing policies described in the language \aopl with respect to inconsistencies, underspecification, ambiguity, and modality conflict. We reified policy rules in order to detect which policy statements cause the particular issue and (if relevant), which fluents of the domain contribute to such problems. In doing so, we defined new properties of \aopl policies and took a special look at what happens at the intersection of authorization and obligation policies.

As part of future work, we plan to create a system that implements this framework in a way that is user-friendly for policy writers and knowledge engineers. We also intend to extend the framework by lifting some of the simplifying restrictions that we imposed here, for instance by studying the case when there is incomplete information about a state or allowing $permitted$ and $obl$ atoms in the conditions of policy rules.

\begin{comment}
\textcolor{red}{Include more tasks from \cite{clmrlb09}; lift restrictions; translation/comparison between our work and \cite{clmrlb09}}
\end{comment}

\bibliographystyle{acmtrans}
\bibliography{iclp2023}

\label{lastpage}
\end{document}